%% file: safety_from_data.tex
\documentclass[letterpaper, 10 pt, conference,usenames,dvipsnames]{ieeeconf}  

\IEEEoverridecommandlockouts                              

\usepackage{graphics} 
\usepackage{epsfig}
\usepackage{epstopdf}
\usepackage{pgfplots} 
\pgfplotsset{compat=newest} 
\pgfplotsset{plot coordinates/math parser=false} 
\newlength\figureheight 
\newlength\figurewidth 

\usepackage{bm}
\usepackage{amsmath,amssymb,amsthm}
\usepackage{mathtools}

\usepackage{algorithm}
\usepackage{algpseudocode}

\usepackage{tikz}\usetikzlibrary{arrows}
\usepackage{subfig}

\usepackage{algorithm}
\usepackage{algpseudocode}
\algrenewcommand\algorithmicindent{1.0em}%
\algtext*{EndFor}
\algtext*{EndProcedure}

\input{definitions/definitionsKim.tex}

\title{\LARGE \bf
	Scalable synthesis of safety certificates from data with application to
	learning-based control
}
\author{Kim P. Wabersich and Melanie N. Zeilinger
\thanks{Kim Wabersich ({\tt\small wabersich@kimpeter.de}) and Melanie N. Zeilinger
({\tt\small mzeilinger@ethz.ch}) are with the Institute for Dynamic Systems and Control,
ETH Zurich, Switzerland. This work was supported by the Swiss National Science Foundation under grant no. PP00P2 157601/1.}
}

\newcommand\copyrighttext{%
	\footnotesize \textbf{Published in: 2018 European Control Conference (ECC), DOI: 10.23919/ECC.2018.8550288.}\\
	\textcopyright 2018 IEEE. Personal use of this material is permitted. Permission from IEEE must be obtained for all other uses, in any current or
	future media, including reprinting/republishing this material for advertising or promotional purposes, creating new collective works,
	for resale or redistribution to servers or lists, or reuse of any copyrighted component of this work in other works.}
\newcommand\copyrightnotice{%
	\begin{tikzpicture}[remember picture,overlay]
		\node[anchor=south,yshift=3pt] at (current page.south) {\fbox{\parbox{\dimexpr\textwidth-\fboxsep-\fboxrule\relax}{\copyrighttext}}};
	\end{tikzpicture}%
}

\begin{document}

\maketitle
\copyrightnotice
\thispagestyle{empty}
\pagestyle{empty}

\begin{abstract}
\noindent
The control of complex systems faces a trade-off between high performance and
safety guarantees, which in particular restricts the application of
learning-based methods to safety-critical systems. 
A recently proposed framework to address this issue is the use of a safety controller,
which guarantees to keep the system within a safe region of the state space.
This paper introduces efficient techniques for the synthesis of a safe set and
control law, which offer improved scalability properties by relying on
approximations based on convex optimization problems.
The first proposed method requires only an approximate linear system model and
Lipschitz continuity of the unknown nonlinear dynamics.
The second method extends the results by showing how a 
Gaussian process prior on the unknown system dynamics
can be used in order to reduce conservatism of the resulting safe set.
We demonstrate the results with numerical examples, including an
autonomous convoy of vehicles.
\end{abstract}


\section{Introduction}

Digitalization opens new perspectives for control engineering and automation
by making large amounts of data from experiments and numerical models available.
Learning-based control exploits this cumulated knowledge and potentially
also performs autonomous exploration of unseen system behavior in order
to find an optimal control policy. An example is deep reinforcement
learning (RL), providing prominent results, one of which is the application
to Atari Arcade video-games \cite{mnih2015human}.

Compared with traditional control techniques, learning-based methods offer the
potential to reduce modeling and controller design effort.
However, many industrial applications are \emph{safety-critical} systems, i.e.
systems with physical constraints that have to be satisfied.
This essentially limits the application of most available learning-based
control algorithms, which do not provide safety certificates.
In order to address this limitation, we present efficient and scalable methods
for the synthesis of a safety strategy consisting of a safe set and corresponding
safe control law, which are cheap to implement and can be applied together
with existing controllers or modern learning techniques to enhance
them with safety guarantees.

\emph{Contributions:}
We consider dynamical systems with \emph{unknown} Lipschitzian nonlinearity and
formulate the safe set and safe controller synthesis as convex optimization
problems, which directly employ available data. Our analysis considers
an approximate linear model of the system and
uses data to incorporate the unknown nonlinear effects. The computations
are based on Lyapunov's method and result in two optimization problems:
The first optimization problem defines a quadratic approximation of the nonlinearity
in the Lyapunov conditions and the second one describes the computation of
the safe set and controller. 
Similar to
\cite{akametalu2014reachability,fisac2017general} the framework can be
used to augment any desired controller, which is lacking safety guarantees,
in particular one which is based on learning.

We extend the technique to reduce conservatism of the safe set by putting
a prior on the unknown dynamics in the form of a Gaussian process model,
which is beneficial, especially in case of high dimensional systems and
sparse data. Due to its less conservative nature, this extension favors
safe exploration beyond the system behavior seen so far and is well
suited for iteratively learning in closed-loop.

We illustrate the approach using examples,
including a convoy of partly non-cooperative autonomous cars.

\emph{Related work:}
Given its relevance in industrial applications, there has been a growing
interest in safe learning methods in the past years.
Extensions of existing RL methods have been developed to enable safe RL
with respect to different notions of safety, see \cite{garcia2015comprehensive}
for a survey. A detailed literature review regarding RL, focusing on safety
with respect to state and input constraints as also considered in this work,
can be found in \cite{fisac2017general}.
There are few results for efficient controller
tuning from data with respect to best worst-case performance (also worst-case
stability under physical constraints) by Bayesian min-max optimization, see
e.g. \cite{wabersich2015automatic}, or by
safety constrained Bayesian optimization as e.g. in
\cite{Berkenkamp2015SafeRobustLearning,berkenkamp2016safequad}.
In \cite{berkenkamp2016safe} a method was developed that
allows to analyze a given closed-loop system
(under an arbitrary RL algorithm) with respect to safety.

Recent developments include the concept of a supervisory framework,
which consists of a safe set in the state space and a safety controller.
As long as the system state is in the interior of the safe set,
any control law (e.g. unsafe learning-based control) can be applied.
The safety controller only interferes if necessary, in case that the system reaches the
boundary of the safe set, see e.g.
\cite{akametalu2014reachability,fisac2017general}.
Such a framework allows for certifying an arbitrary learning-based
control algorithm with safety guarantees.
Previously proposed techniques are based on a differential game formulation,
which results in solving a min-max optimal control problem.
An active field of research aims at extending these techniques to larger scale
systems, mostly by considering special cases as described, e.g., in \cite{chen2015safe,kaynama2015scalable,fisac2015pursuit}.
For some relevant cases, the existence of analytic
solutions \cite{darbon2016algorithms} has been shown.
The results presented in this paper are based on the concept of a safety
framework, but compared to previous work we focus on approximation techniques 
to improve scalability with respect to the system dimension.

\emph{Structure of the paper:}
In Section
\ref{sec:problem_description} we state the problem and
in Section \ref{sec:safe_sets_from_data} we present our main
result for safe set and controller computation. We then
show an extension using a stronger assumption
on the unknown system dynamics by considering
Gaussian processes in Section \ref{sec:safe_active_exploration_for_nonlinear_systems}
in order to reduce conservatism of the safe set.
The results are demonstrated on numerical examples
within the respective sections.
We conclude the paper in Section \ref{sec:conclusion}.

\emph{Notation:}
The set of symmetric matrices of dimension $n$
is $\mSetSymMat{n}$, the set of positive (semi-)
definite matrices is ($\mSetPosSemSymMat{n}$) $\mSetPosSymMat{n}$, 
the set of integers in the interval $[a,b]\subset\RR$ is
$\mIntInt{a}{b}$, the set of integers in the interval
$[a,\infty)\subset\RR$ is $\mIntGeq{a}$, and for $\epsilon > 0$
let
$\mBall{\epsilon}{\bar x} = \left\lbrace x \in \RR | \mNorm{x-\bar x}{2}\leq \epsilon\right\rbrace$. The boundary of an arbitrary compact set $\mathcal C \subset \RR^n$
is $\partial \mathcal C$.
Given a set $\mDataSet = \lbrace (x_i, y_i) \rbrace_{i=1}^N$, let
$\mDataSetX=\lbrace x_i \rbrace_{i=1}^{N}$ and
$\mDataSetY=\lbrace y_i \rbrace_{i=1}^{N}$. Define
$\mDefFunction{\mGridFunction}{\RR^n}{\mDataSetX}$ as
$\mGridFunction(x)={\argmin_{\bar x\in\mDataSetX}\mNorm{\bar x - x}{2}}$,
which picks the closest element in $\mDataSetX$ with
respect to $x\in\RR^n$ under the $2$-norm.
Given a set $\mDataRegion\subset\RR^n$ and a locally Lipschitz
continuous function $\mDefFunction{f}{\RR^n}{\RR^m}$,
the local Lipschitz constant
$L \leq |f(x) - f(y)|/\mNorm{x-y}{2}$ for all
$x,y\in\mDataRegion$ is denoted by $L_{f(x)}(\mDataRegion)$.
The Minkowski sum of two sets $\mathcal A_1, \mathcal A_2 \subset \RR$
is denoted by $\mathcal A_1 \oplus \mathcal A_2$.


\section{Problem description}\label{sec:problem_description}
\begin{figure}[t]
	\centering
	\input{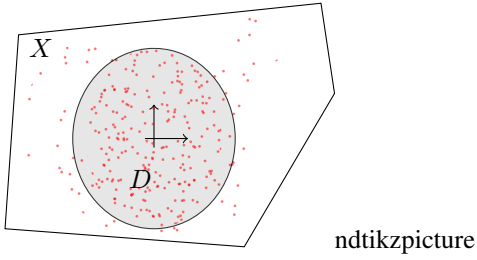}
	\caption{Illustration of Assumption~\ref{ass:big_data}. Red dots display observations
	$(x_i,d(x_i))$.}\label{fig:illustration_Ass_big_data}
\end{figure}
We consider deterministic nonlinear systems of the form
\begin{align}\label{eq:general_system}
  \dot x(t) = Ax(t) + Bu(t) + d(x(t))
\end{align}
where $A\in\RR^{n\times n}$, $B \in \RR^{n\times m}$ and
$d: \RR^{n} \rightarrow \RR^n$ is locally Lipschitz continuous.
The system is subject to polytopic state constraints $x(t) \in \XX
:= \lbrace x\in \RR^n|A_x x \leq b_x \rbrace$, 
$A_x\in\RR^{n_x\times n}$, $b_x\in\RR^{n_x}$ and polytopic input
constraints $u(t)\in\UU:=\lbrace u\in\RR^m|A_u u \leq b_u \rbrace$,
$A_u\in\RR^{n_u\times m}$, $b_u\in\RR^{n_u}$.
The origin is contained in $\XX$, $(A,B)$ is controllable, and
the system state is fully observable.

The explicit form of the nonlinearity $d$ is unknown, hence $d$ is assumed to be
a memoryless black-box function, which will be identified from system measurements.
However, we assume that $B$, i.e. the influences
of the control input, are known. The matrix $A$ incorporates
system knowledge in form of a linear model, which will be used in
the design procedure in Section \ref{subsec:choice_of_P}. The linear model
can, e.g., be selected as an approximate system model.
For identification of $d(x)$, we have access to finitely many observations
$\mDataSet=\lbrace(x_i,d(x_i))\rbrace_{i=1}^{N}$ that
fulfill the following property (Figure~\ref{fig:illustration_Ass_big_data}).
\begin{assumption}\label{ass:big_data}
	Given a set $\mDataSet=\lbrace (x_i, d(x_i)) \rbrace_{i=1}^N$
	with $N$ data tuples, where $x_i\in\XX$, there exists
	a non-trivial subset $\mDataRegionDense \subseteq \XX$
	such that for any $x\in\mDataRegionDense$ there exists an $x_i\in\mDataSetX$
	such that $||x-x_i||_2 \leq \delta$.
\end{assumption}
\begin{remark}
	Assumption~\ref{ass:big_data} implies that there exists a
	region, in which the collected data samples are dense.
	Intuitively, one can think of $\delta$ together with the
	Lipschitz constant $L$, as a `measure of knowledge'
	that we have about $d(x)$ inside $\mDataRegionDense$. The `knowledge'
	increases as $\delta$ gets smaller.
\end{remark}
\begin{remark}\label{rem:noisy_data}
	For simplicity, we assume noise-free data $\mDataSet$.
	It would, however, be possible to incorporate bounded or stochastic
	noise with only minor changes.
\end{remark}
We consider the problem of providing a safety certificate
for an arbitrary control law by means of a safe set and controller,
as proposed in \cite{akametalu2014reachability,fisac2017general}.
Consider a potentially unsafe control strategy $\bar u(t)$, obtained for example
by application of RL (which often cannot
guarantee constraint satisfaction). In order to
achieve minimal interference with the desired control $\bar u(t)$,
the goal is to compute a set of states $\mSafeSet$, for which we know 
a control strategy $u_\mSafeSet(t)$ such that input and state constraints
will be satisfied for all future times, in particular considering
that $d(x)$ is unknown. The control $\bar u(t)$ can then safely be applied
in the interior of S, until it becomes necessary to take a safety-ensuring
action $u_\mSafeSet$ on the boundary of $\mSafeSet$ which guarantees 
that we stay in $\mSafeSet$, i.e. that we can still provide a safe control strategy
in the future. More formally:

\begin{definition}\label{def:safe_set}
	A set $\mSafeSet\subseteq \XX$ is called a \emph{safe set} for
	system \eqref{eq:general_system} if there exists a
	\emph{safe control law} $\mDefFunction{u_\mSafeSet}{\mSafeSet}{\UU}$
	such that for an arbitrary (learning-based) policy $\bar u(t)$,
	the safe controller
	\begin{align}\label{eq:safe_control_law}
		u(t) = 
			\begin{cases}
				u_S(x(t)),~&x(t)\in\partial\mSafeSet \lor \bar u(t) \notin \UU\\
				\bar u(t),~&\text{otherwise}
			\end{cases}
	\end{align}
	guarantees that the state $x(t)$ is contained in $\mSafeSet$ for all $t>0$ if
	$x(0)\in \mSafeSet$.
\end{definition}

In particular, we aim at finding an algorithm that
scales well in computational complexity with respect to
the dimensionality of system \eqref{eq:general_system}
as well as the number of measurements.



\section{Safe-sets for nonlinear systems from data}\label{sec:safe_sets_from_data}
We first introduce the class of safe-sets considered. Afterwards, we motivate
the proposed method and highlight its basic idea using an example,
in order to then introduce the algorithm for safe set and controller computation
in the remainder of the section.

\subsection{Ellipsoidal safe set}

In order to provide a scalable optimization-based approach, we restrict 
the form of the safe set to an ellipsoidal set of the form
\begin{align}\label{eq:quadratic_safe_set}
  \mSafeSet^P(\gamma) = \left\lbrace x\in\XX | x^\top P x \leq \gamma \right\rbrace
\end{align}
with $P\in\mSetPosSymMat{n}, \gamma\in\RR^n, \gamma >0$, and the safe controller
to the class of linear state feedback control laws $u_\mSafeSet=Kx$ with
$K\in \RR^{m\times n}$.
To construct $\mSafeSet^P(\gamma)$, we leverage Lyapunov's direct method,
with a quadratic Lyapunov function
$V(x)=x^\top (\gamma^{-1}P) x$.
By standard Lyapunov arguments, the following sufficient conditions
ensure, analogously to \cite[Lemma 1]{akametalu2014reachability},
that $\mSafeSet^P(\gamma)$ fulfills Definition~\ref{def:safe_set}:
\begin{subequations}
\begin{align}
\label{eq:nominal_invariant_set_req_X}
  \mSafeSet^P(\gamma) &\subseteq \XX\\\label{eq:nominal_invariant_set_req_U}
  Kx&\in \UU \\\label{eq:nominal_invariant_set_req_vdot}
  \qquad \dot V(x) &\leq 0
\end{align}
\end{subequations}
for all $x\in\partial\mSafeSet^P(\gamma)$.

\subsection{Motivating Example}
Consider the system $\dot x = x +d(x) + u$ with $d(x)=- x^3$
subject to input constraints $|u|\leq 2$ and state constraints
$|x|\leq 2$. The task is to find a \emph{safe} interval $\mSafeSet=[a,b]$
and a corresponding \emph{safe control law}
$\mDefFunction{u_\mSafeSet}{[-2,2]}{[-2,2]}$ according to Definition~\ref{def:safe_set}.
The nonlinearity $d(x)$ is unknown, but we are
given a set of noise-free observations $\mDataSet=\lbrace x_i, d(x_i) \rbrace_{i=1}^N$
such that the convex hull $\mathrm{conv}(\lbrace x_i \rbrace_{i=1}^N )$
equals the state space $[-2,2]$ (this will not be a necessary assumption later,
see also Assumption~\ref{ass:big_data}). We consider a linear state feedback $u_\mSafeSet=kx$, $k\in \RR$.

\textit{Robust approach}: Without any knowledge of $d(x)$, a robust approach is
to consider the dynamics $\dot x = x + w + u$, with $|w|\leq 8$, where the bound on $w$ is
estimated from $\mDataSet$. In this case, there does not exist a
feasible controller gain k w.r.t. input constraints for any state
and $\forall |w|\leq 8$, such that $\dot x \leq 0$, i.e. there does not
exist a safe set.

\textit{Proposed approach}: Let $V(x)=px^2,~p>0$ be our Lyapunov candidate function.
We analyze $\dot V(x) = 2p(x^2+k x^2+ x d( x))\leq 0$.
At the boundary of the state space we have the measurements $2d(2)=-2d(-2)=-16$,
providing that $\dot V (x) \leq 0$. 
By standard Lyapunov arguments, this implies that for $k=0$ we have that
for all $x(0)\in \partial \mSafeSet=\lbrace -2, 2\rbrace$,
$x(t)\in \mSafeSet$ for all $t>0$.
We conclude that $u_\mSafeSet(t)=0$ is a safe controller for
which the state constraint set constitutes a safe set.

This example highlights that rather than taking uniform bounds
on the unknown dynamics, we can provide less conservative safe sets
by quantifying the effect of the unknown dynamics in the form of
\emph{state-dependent} disturbances, which can be inferred from the
available data $\mDataSet$.
In the following, we will exploit this concept for safe set and controller
computations and conclude the section with two examples.

\subsection{Computation of the safe set and controller}
Given a matrix $P$ (see Section \ref{subsec:choice_of_P})
that determines the shape of the safe set, we write the
problem of finding the size of the safe set $\mSafeSet^P(\gamma)$
and a corresponding safe controller $u_\mSafeSet$ as two
consecutive convex optimization problems.
\begin{remark}\label{rem:ellipsoidal_data_region}
	Given the shape $P$ of the safe set \eqref{eq:quadratic_safe_set},
	there exists a $\bar \gamma >0$ small enough such that
	Assumption~\ref{ass:big_data} is satisfied on the sub-level
	set $\mSafeSet^P(\bar \gamma)$, i.e.
	\begin{align}\label{eq:ellipsoidal_data_region}
		\mSafeSet^P(\bar\gamma)\subseteq \mDataRegionDense.
	\end{align}
\end{remark}
The proposed procedure first uses data to bound the effects of
the nonlinearity on the Lyapunov decrease by a quadratic form
in the largest possible safe set, i.e. over $\gamma\in (0,\bar \gamma]$.
This quadratic bound is then used
as input to the second optimization problem, which computes
the controller and set size in order to take into account the nonlinearities
in addition to the linear system dynamics. The restriction to a quadratic bound
of the nonlinearity is motivated by the fact 
that it can be treated efficiently by means of a convex
problem.

In order to reduce conservatism, we bound the nonlinearity 
on sub-regions of the safe set, described by intervals
\begin{align}\label{eq:intervals}
	\gamma\in\Gamma_i=[\gamma_1^i, \gamma_2^i],~\gamma_1^i<\gamma_2^i,~ \gamma_2^i\leq \bar \gamma,~i=1,2,..,n_\Gamma,
\end{align}
which are defined such that $\mSafeSet^P(\gamma)\subseteq \mDataRegionDense$
for any $\gamma\in\Gamma_i$.
Note that the selection of sub-intervals is possible as we will use the quadratic bound for
upper bounding \eqref{eq:nominal_invariant_set_req_vdot}, which is only
required to hold on $\partial \mSafeSet^P(\gamma)$.
For every interval $\Gamma_i$, we then formulate
two convex optimization problems in order to determine the volume of the safe
set and the safe controller. In case no solution exists,
the interval can be reduced. In general, the smaller the intervals are chosen,
the less conservative the bound will be.

\emph{Bounding the nonlinear effects:} 
Given an interval $\Gamma_i$, consider
the neighborhood
$\mSafeRing(\Gamma_i)=
\lbrace \mSafeSet^P(\gamma_i^2)\setminus
\mSafeSet^P(\gamma_i^1)\rbrace \oplus \mBall{\delta}{0}$.
The indices of data samples inside the set $\mSafeRing(\Gamma_i)$
are given by $\mSafeRingIndices(\Gamma_i) = 
\lbrace k \in \mIntGeq{1}|x_k\in\mDataSetX,~ x_k\in \mSafeRing(\Gamma_i)\rbrace$.
We seek to find a quadratic bound on the nonlinearity arising in the
Lyapunov decrease \eqref{eq:nominal_invariant_set_req_vdot} for all
$\bar x \in \mSafeSet^P(\gamma_i^2)\setminus\mSafeSet^P(\gamma_i^1)$, i.e.
to find a $\mQuadBoundMatrixFunction{\Gamma_i}$ such that
\begin{align}\nonumber
	 \dot V(\bar x)=&2\gamma^{-1} {\bar x}^\top P(A+BK){\bar x} + 2\gamma^{-1} {\bar x}^\top Pd({\bar x}) \\\label{eq:general_opt_proof_2}
	 \leq &2\gamma^{-1}{\bar x}^\top P(A+BK){\bar x} + 2\gamma^{-1}{\bar x}^\top \mQuadBoundMatrixFunction{\Gamma_i} {\bar x}.
\end{align}
The first optimization problem for bounding the nonlinearity over each interval is
given by 
\begin{subequations}\label{eq:quad_bound}
	\begin{align}
		\mQuadBoundMatrixFunction{\Gamma_i} = &\argmin_{\tilde \mQuadBoundMatrix \in \mSetSymMat{n}} \sum_{k \in \mSafeRingIndices(\Gamma)} \left(x_k^\top
		\tilde \mQuadBoundMatrix x_k - p_k\right)^2\\\nonumber
		\text{s.t.}&\text{ for all } k\in \mSafeRingIndices(\Gamma_i):
		\\\label{eq:quad_bound_2}
		&~\lambda_k \geq 0
		\\\label{eq:quad_bound_3}
		&~
			\begin{pmatrix}
				-\tilde \mQuadBoundMatrix - \lambda_k I_n & \lambda_k x_k\\ \lambda_k x_k^\top & -\lambda_k\left(x_k^\top x_k - \delta^2\right) + p_k
			\end{pmatrix}\preceq 0
	\end{align}
\end{subequations}
with $p_k=x_k^\top P d(x_k) + \delta L_{x^\top P d(x)}(\mSafeRing(\Gamma_i))$
and $\delta$ as defined in Assumption~\ref{ass:big_data}.
\begin{lemma}\label{lem:quad_bound}
	Let Assumption~\ref{ass:big_data} hold and let $\bar \gamma$
	satisfy \eqref{eq:ellipsoidal_data_region}. Consider an
	interval $\Gamma_i$ according to \eqref{eq:intervals}.
	If \eqref{eq:quad_bound} attains a solution, then for all $\bar x \in
	\mSafeSet^P(\gamma_i^2)\setminus\mSafeSet^P(\gamma_i^1)$ it holds that
	\begin{align}\label{eq:quad_bound_result}
		\bar x^\top  P d(\bar x) \leq \bar x^\top \mQuadBoundMatrix (\Gamma_i) \bar x.
	\end{align}
\end{lemma}
\begin{proof}
	We prove that for all $\bar x \in
	\mSafeSet^P(\gamma_i^2)\setminus\mSafeSet^P(\gamma_i^1)$
	the implication
	$\eqref{eq:quad_bound_2},\eqref{eq:quad_bound_3}\Rightarrow \eqref{eq:quad_bound_result}$ holds.
	First note that for any $\bar x \in\mSafeSet^P(\gamma_i^2)\setminus\mSafeSet^P(\gamma_i^1)$ there exists
	a $k\in \mSafeRingIndices(\Gamma_i)$ such that $||\bar x-x_k||\leq \delta$
	by the definition of the intervals, i.e. $\gamma_i^2\leq \bar \gamma$, see also
	Remark \ref{rem:ellipsoidal_data_region}. For notational ease let
	$f(\bar x) = d(\bar x)^\top P \bar x$. Equation \eqref{eq:quad_bound_result} reads
	$f(\bar x)-\bar x^\top\mQuadBoundMatrix(\Gamma_i)\bar x\leq 0 $.
	For all $k\in\mSafeRingIndices(\Gamma_i)$ and for all $\bar x_k \in \mBall{\delta}{x_k}$
	we have therefore by Lipschitz continuity
	$f(\bar x_k) -  f(\mGridFunction{(\bar x_k)})+ f(\mGridFunction{(\bar x_k)}) - \bar x_k^\top \mQuadBoundMatrix(\Gamma_i)\bar x_k
	\leq p_k - \bar x_k^\top \mQuadBoundMatrix(\Gamma_i)\bar x_k$.
	Note that by the definition of the intervals and Remark \ref{rem:ellipsoidal_data_region} the relation
	$\lbrace\mSafeSet^P(\gamma_i^2)\setminus\mSafeSet^P(\gamma_i^1)\rbrace\subset\bigcup_{i\in\mSafeRingIndices(\Gamma_i)}\mBall{\delta}{x_k}$ holds.
	As a consequence, if for all $k\in\mSafeRingIndices(\Gamma_i)$ and for all
	$\bar x_k \in \mBall{\delta}{x_k}$ we have that $p_k - \bar x_k^\top \mQuadBoundMatrix(\Gamma_i)\bar x_k\leq 0$,
	then the quadratic bound \eqref{eq:quad_bound_result} holds for all
	$\bar x\in\mSafeSet^P(\gamma_i^2)\setminus\mSafeSet^P(\gamma_i^1)$. Finally, using the
	S-Lemma (see \cite{polik2007survey}) the condition
	$
		\bar x \in \mBall{\delta}{x_k}
		\Rightarrow p_i-\bar x^\top Q(\gamma)\bar x \leq 0
	$
	is equal to \eqref{eq:quad_bound_2},\eqref{eq:quad_bound_3} which completes the proof.
\end{proof}
\begin{remark}\label{rem:large_data_set_quad_bound}
The optimization problem in \eqref{eq:quad_bound} is a convex semidefinite programming problem.
In case that there are more observations ($N\gg0$) than the optimization algorithm can handle in \eqref{eq:quad_bound},
one can iteratively calculate $\mQuadBoundMatrixFunction{\Gamma_i}$: Solve
\eqref{eq:quad_bound} using a subset of $\mDataSet$ in order to obtain
$\mQuadBoundMatrixFunctionIteration{1}{\Gamma_i}$, in the next iteration choose another disjoint subset of
$\mDataSet$ and add the constraint $\tilde \mQuadBoundMatrix\succeq \mQuadBoundMatrixFunctionIteration{1}{\Gamma_i}$
to \eqref{eq:quad_bound} in order to obtain $\mQuadBoundMatrixFunctionIteration{2}{\Gamma_i}$. Repeat
until all subsets of $\mDataSet$ are processed which yields a feasible, possibly suboptimal
solution to \eqref{eq:quad_bound}.
\end{remark}
Problem (8) provides a  bound on the nonlinear effect in the Lyapunov decrease
by means of the Lipschitz constant of $x^\top Pd(x)$. In practice, a local Lipschitz constant can
e.g. be obtained from data as
$\hat L_{x^\top P d(x)}(\mSafeRing(\Gamma_i)) = 2\max_{k_1,k_2\in\mSafeRingIndices(\Gamma_i)} ||x_{k_1}^\top Pd(x_{k_1})-x_{k_2}^\top Pd(x_{k_2})||/||x_{k_1} - x_{k_2}||$.

Since $Q(\Gamma_1)$ does not have to be positive definite,
stabilizing effects of the nonlinearity can be considered in \eqref{eq:general_opt_proof_2},
i.e. $\bar x^\top Q(\Gamma_i) \bar x$ can be negative and therefore contribute
to rendering $\dot V(x)$ negative on the boundary of the safe set.
This is demonstrated in Section~\ref{sec:simple_example_lipschitz}, see also
Figure~\ref{fig:example} (right). \medskip

\emph{Calculation of safe level set and safe controller:} By using the
quadratic bound from the first optimization problem, we are
able to state the second optimization problem for safe set
and controller design satisfying the conditions
\eqref{eq:nominal_invariant_set_req_X}-\eqref{eq:nominal_invariant_set_req_vdot}
as follows. Let $\gamma \in \RR$, $E = P^{-1}$, $Y \in \RR^{m\times n}$, 
$\Gamma_i$ be chosen according to \eqref{eq:intervals}. The optimization problem
is given by
\begin{subequations}\label{eq:general_opt}
\begin{align}
	\min_{\gamma,Y}  & -\gamma \\\label{eq:general_opt_0}
	\mathrm{s.t.} &~ \gamma\in \Gamma_i\\\label{eq:general_opt_1}
	&~ AE\gamma + EA^T\gamma + BY + Y^\top B^\top + 2E \mQuadBoundMatrixFunction{\Gamma_i} E\gamma \preceq 0\\\nonumber
	&~\forall j\in\mIntInt{0}{n_x}:\\\label{eq:general_opt_5}
	&~~~~\begin{pmatrix}b_{x,j}^2 & A_{x,j}E \\ EA_{x,j}^\top & \gamma E\end{pmatrix}\succeq 0\\\nonumber
	&~\forall k\in\mIntInt{0}{n_u}: \\\label{eq:general_opt_6}
	&~~~~\begin{pmatrix}b_{u,k}^2 & A_{u,k}Y \\ Y^\top A_{u,k}^\top & \gamma E\end{pmatrix}\succeq 0.
\end{align}
\end{subequations}

\begin{theorem}\label{thm:invariant_set_from_data_lipschitz}	
	Let Assumption~\ref{ass:big_data} hold and let $\bar \gamma$
	satisfy \eqref{eq:ellipsoidal_data_region}. Consider an	interval
	$\Gamma_i$ according to \eqref{eq:intervals}. If \eqref{eq:general_opt}
	attains a solution $\lbrace \gamma^*, Y^*\rbrace$, then
	$\mSafeSet^P(\gamma^*)$ is a safe set for system \eqref{eq:general_system}
	according to Definition~\ref{def:safe_set} with $u_S(x) = Kx$,
	$K={\gamma^*}^{-1}Y^*E^{-1}$.
\end{theorem} 
\begin{proof}
	We prove the result in two steps:
	1) Conditions \eqref{eq:general_opt_5}-\eqref{eq:general_opt_6} imply
	(\ref{eq:nominal_invariant_set_req_X}) and (\ref{eq:nominal_invariant_set_req_U}).
	By \cite[Section 5.2.2]{boyd1994linear} we can rewrite (\ref{eq:nominal_invariant_set_req_U})
	as $A_{u,i}K(\gamma^{-1}P)^{-1}K^\top A_{u,i}^\top \leq b_{u,i}^2$ which equals
	\eqref{eq:general_opt_6}. The matrix inequality for the states can be derived similarly.
	2) Conditions \eqref{eq:general_opt_0}-\eqref{eq:general_opt_1} ensure
	that $\mSafeSet^P(\gamma^*)$ fulfills \eqref{eq:nominal_invariant_set_req_vdot}.
	For all $x \in \partial \mSafeSet^P(\gamma)$ we have to fulfill
	\eqref{eq:nominal_invariant_set_req_vdot} which is implied by \eqref{eq:general_opt_proof_2}.
	A sufficient condition for \eqref{eq:nominal_invariant_set_req_vdot} is therefore
	$\gamma^{-1}P(A+BK) + \gamma^{-1}(A+BK)^\top P + 2\gamma^{-1} \mQuadBoundMatrixFunction{\Gamma_i} \preceq 0$,
	i.e. that $\dot V(x)\leq 0$ for all $x\in\RR^n$.
	Multiplying from left and right by $\gamma P^{-1}$ yields
    \eqref{eq:general_opt_1}.
	We have shown that $\mSafeSet^P(\gamma^*)$ is a safe set according to Definition~\ref{def:safe_set}.
	The objective in \eqref{eq:general_opt} yields the largest safe set under these sufficient
	conditions.
\end{proof}
Note that optimizing over $P$ and $K$ in \eqref{eq:general_opt} is not
possible, because the bound obtained in the first optimization
step depends on $P$.

Problem \eqref{eq:quad_bound} and \eqref{eq:general_opt} provide (semidefinite) convex optimization problems for
computing a safe set and controller by directly employing data points. Such problems can be solved
efficiently even for higher dimensions, see e.g. \cite{toh1999sdpt3,mosek}.
While this offers a general approach with favorable scalability properties, the limitation is that the resulting safe set cannot be larger than the convex hull of the data points plus
a $\delta$-neighborhood. Exploration is therefore limited.
Nevertheless, initially collected data can be iteratively extended inside $\mDataRegionDense$
such that $\delta$ from Assumption~\ref{ass:big_data} gets smaller over time.
Recomputation of the safe set can then reduce conservatism.
We present an extension in Section \ref{sec:safe_active_exploration_for_nonlinear_systems}
which further reduces conservatism and improves exploration. 

\subsection{Shape of the safe set}\label{subsec:choice_of_P}
Using the linear model of the system dynamics in \eqref{eq:general_system},
we define an approximate initial shape matrix $P$ of the safe set by
neglecting the unknown nonlinearity. Assume that
$\mDataRegionDense$ is given by $\lbrace x\in \RR^n | x^\top A_\delta x\leq 1 \rbrace$
with $A_\delta\in\mSetPosSymMat{n}$, which can be e.g. calculated
as the minimum volume covering ellipse of the data points $\mDataSet$,
see \cite[p. 222]{boyd2004convex}. We can find a safe
set for system \eqref{eq:general_system} by setting $d(x)=0$,
resulting in the following optimization problem with
$E \in \mSetPosSymMat{n}$ and $Y \in \RR^{m\times n}$
\begin{subequations}\label{eq:intial_guess_P}
	\begin{align}
		\min_{E,Y}  & -\mathrm{log det} (E) \\\label{eq:intial_guess_P_0}
		\mathrm{s.t.} & ~~A_\delta^{-1}\succeq E \\\label{eq:intial_guess_P_1}
		& ~~AE + EA^T + BY + Y^\top B^\top \preceq 0\\
		& ~~(\ref{eq:general_opt_5}),(\ref{eq:general_opt_6}).
	\end{align}
\end{subequations}
If \eqref{eq:intial_guess_P} attains a solution, then analogously to Theorem~\ref{thm:invariant_set_from_data_lipschitz},
we obtain a safe set for system \eqref{eq:general_system} according to Definition~\ref{def:safe_set}
with $P=E^{-1}$, $u_S(x) = Kx$, and $K=YE^{-1}$, but for $d(x)=0$. Constraint \eqref{eq:intial_guess_P_0}
ensures that the safe set is a subset of $\mDataRegionDense$, i.e. the set in which
Assumption~\ref{ass:big_data} is satisfied and therefore the data is dense.
This implies that $(0,1]$, i.e. $\bar \gamma=1$, is the maximum set size
such that Assumption~\ref{ass:big_data} as well as all constraints are fulfilled.
The optimization problem \eqref{eq:general_opt} then aims at improving
the initial approximation obtained via \eqref{eq:intial_guess_P} with respect
to the nonlinearity by designing a different control law and safe set size.

\subsection{Illustrative numerical example}\label{sec:simple_example_lipschitz}
\begin{figure}[t]
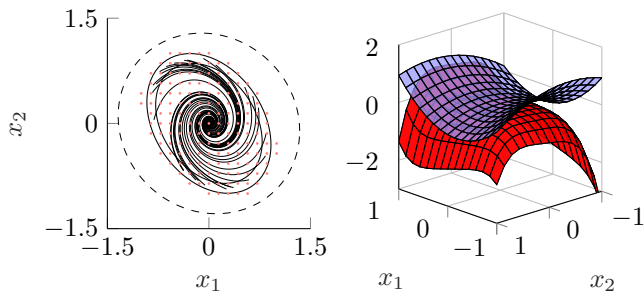
%
    \centering
    \subfloat{\input{fig/tkiz/simpleExample.tkiz}}%
    \subfloat{\input{fig/tkiz/simpleExampleNonlinearBound.tkiz}}%
    \caption{\textbf{Left:} Sample trajectories under the safe control law, when starting on the
		boundary of the safe set $\mSafeSet^P$. Red dots: Data points, black lines: Closed-loop system
		trajectories, elliptic ring $\mSafeRing({\Gamma_1})$, which contains $\partial\mSafeSet^P$.
		Dashed set: Safe set using $Q_{GP}$ from Section
		\ref{sec:safe_active_exploration_for_nonlinear_systems}.
		\textbf{Right:} Quadratic bound $x^\top \mQuadBoundMatrixFunction{\Gamma_i} x$ according to
		Lemma~\ref{lem:quad_bound} (blue) of non-linear term $x^\top P d(x)$ (red).}%
    \label{fig:example}%
\end{figure}

Consider a nonlinear system of the form \eqref{eq:general_system}, where
$A = \begin{psmallmatrix} -1 &2 \\ -3& 4 \end{psmallmatrix}$,
$B=\begin{psmallmatrix} 0.5 \\ -2 \end{psmallmatrix}$,
$d(x) = \begin{psmallmatrix} 0.5x_1^4 \\ 0.35-1.5x_2^3 \end{psmallmatrix}$
with input constraints $|u|\leq 3$ and state constraints $|x|\leq 2$.
The origin is not a stable equilibrium, and is neither an equilibrium point.
For simplicity we consider a grid of data points as illustrated in Figure~\ref{fig:example},
however any data set fulfilling Assumption~\ref{ass:big_data} could be used. The data
was taken inside the set $\lbrace x\in\RR^n|x^\top P x \leq 1 \rbrace$ with
$P = \left(\begin{smallmatrix} 0.7651 & 0.2162 \\ 0.2162 & 0.6481 \end{smallmatrix}\right)$
obtained by solving \eqref{eq:intial_guess_P},
which also corresponds to $\mDataRegionDense$ with $\delta=0.15$.
We apply Theorem~\ref{thm:invariant_set_from_data_lipschitz} by
solving \eqref{eq:quad_bound} and \eqref{eq:general_opt} for $\Gamma_1=[0.9,1]$,
$L_{x^\top P d(x)}(\mSafeRing(\Gamma_1)) \approx 6.02$ and obtain
$ K^* = (0.5261,~2.2953),\gamma^*=1$. The results are illustrated in Figure
\ref{fig:example}. Note that in general the quadratic bound must only hold on
$\partial\mSafeSet(\gamma)$ and could therefore be violated around the origin.

\subsection{Simulation: Safety for autonomous convoys}
Consider a convoy of cars or trucks as depicted in Figure~\ref{fig:car_convoy}.
Given a target velocity $v_{\text{tar}}$ and a possibly small target distance $x_{\text{tar}}$,
the goal is to drive closely behind each other, in order to leverage slipstream effects for efficiency.
We assume that it is possible to overwrite the local controllers, i.e. the acceleration of
car~1, car~3 and car~4 in a centralized way if necessary to ensure safety.
During a supervised observation phase, initial data about
the system is collected. We consider the problem
of finding a safe, centralized control law, and a safe set
such that the cars will not crash,
even if we cannot determine the acceleration of car~2 and car~5.

Let $z_{i+1\rightarrow i}=x_{\text{tar}}-x_{i+1\rightarrow i}$ be the difference
between the target distance $x_{\text{tar}}$ and the actual distance $x_{i+1\rightarrow i}$ 
of car $i+1$ and car $i$. Let $v_i$ be the difference between the target velocity of
the convoy $v_{\text{tar}}$ and the velocity $\bar v_i$ of car $i$.
The dynamics of all cars $i=1,...,5$ are given as
$\dot z_{i+1\rightarrow i} = v_{i+1} - v_i$, $\dot v_{i} = u_i$
where $u_i$ is the applied acceleration. The control law of
car $1$ is given by $u_1(v_1) = -v_1$, of cars $3,4$ by
$u_i(z_{i\rightarrow i-1},v_i) = 0.1z_{i\rightarrow i-1}-0.3v_i$ and
of cars $2,5$ (which we cannot overwrite) by
$u_2(z_{2\rightarrow1},v_2) = \max{\lbrace\min{\lbrace z_{2\rightarrow1}-v_2,0.9\rbrace},-0.9}\rbrace$,
$u_5(z_{5\rightarrow 4},v_5) = \max{\lbrace\min{\lbrace z_{5\rightarrow 4}-v_5,0.9\rbrace},-0.9}\rbrace$,
i.e. they apply a saturated, stabilizing state feedback law and are therefore nonlinear.
The target distance between the cars is 1 meter. In order to avoid a crash the
state constraints are given by $x_{i+1 \rightarrow i}\geq 0$. The maximum acceleration
of cars 1, 3, and 4 is $3~[m/s^2]$, i.e. $|u_i|\leq 3,~ i=1,3,4$.
We are given observations of $\dot z_{2\rightarrow 1},\dot v_2$ and
$\dot z_{5\rightarrow 4}, \dot v_5$ in the interval $[-0.8, 0.8]$ with $\delta = 0.013$.

In Figure~\ref{fig:car_convoy_example}, a numerical simulation under the
resulting safe control law \eqref{eq:safe_control_law} is shown, starting
from the boundary of the safe set with $v_2(0)=0.02~[m/s],x_{2\rightarrow 1}(0)=-0.76~[m]$
and the remaining states equal to zero, which represents the situation that the
second car is too close to the first one and its velocity is slightly higher than the reference
velocity.
As we can see in Figure~\ref{fig:car_convoy_example}, the first car has to accelerate quickly
several times during the first two seconds for safety reasons, since the second car (which cannot be controlled)
would decelerate more as car $3$ would be able to compensate. The same situation occurs at $4.2$ seconds between car $3$ and car $4$ and
at around $10$ seconds between car $1$ and car $2$ again.
After six safety interventions in total, the local controllers of the cars
are able to stabilize the overall system.

\begin{figure}[t]
	\vspace{0.2cm}
	\centering
	\input{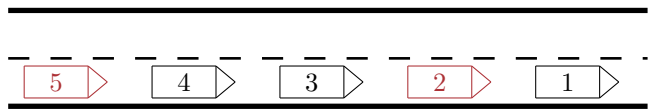}
	\caption{Illustration of the autonomous car convoy. 
	The acceleration of red cars cannot be controlled.}
	\label{fig:car_convoy}
\end{figure}

\begin{figure}[t]
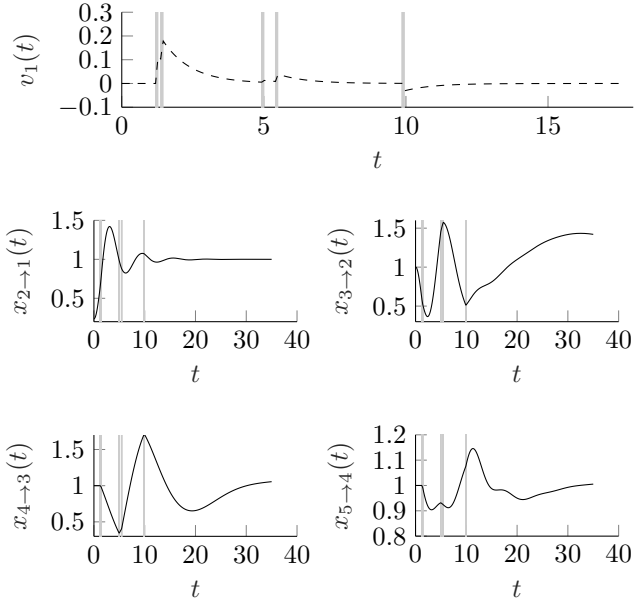

	\centering
	\subfloat{\input{fig/tkiz/FiveCarConvoy1.tkiz}}\\
    \subfloat{\input{fig/tkiz/FiveCarConvoy2.tkiz}}%
    \subfloat{\input{fig/tkiz/FiveCarConvoy3.tkiz}}\\
    \subfloat{\input{fig/tkiz/FiveCarConvoy4.tkiz}}%
    \subfloat{\input{fig/tkiz/FiveCarConvoy5.tkiz}}
	\caption{Sample trajectory of the car convoy
	under the safety framework, when starting on the boundary of $\mSafeSet^P(\gamma)$.
	Grey lines indicate times when the safe control law $u_\mSafeSet$ is applied.}
	\label{fig:car_convoy_example}
\end{figure}

\section{Safe-sets using Gaussian Processes}
\label{sec:safe_active_exploration_for_nonlinear_systems}
The previous sections are based on Lipschitz continuity
of the unknown nonlinearity $d(x)$, which was
incorporated into the quadratic upper bound in \eqref{eq:general_opt},
see Lemma~\ref{lem:quad_bound}.
A shortcoming of using only Lipschitz continuity as `prior' knowledge is the
requirement of relatively dense and structured data, see
Assumption~\ref{ass:big_data}. This implies that
almost no `exploration' can be made beyond the data, observed so far.
By putting a stronger prior on the class of functions for modeling the
nonlinearity $d(x)$, we develop a less conservative quadratic bound,
which can then be used in \eqref{eq:general_opt} and is generally expected to yield a
larger safe set. In addition, it improves exploration beyond the data points allowing to
iteratively improve the safe set during closed-loop operation, where initially few data
points are generally available.

\subsection{Gaussian Processes}
We use a Gaussian process model (GP) in order to perform
Bayesian inference on the unknown nonlinearity
(see. e.g. \cite{rasmussen2006gaussian}) for each element $d_i(x)$ of $d(x)$.
The GP is defined by a mean function $\mu^i(x)$, together with a
covariance (kernel) function $k^i( x,  x')$, denoted with
$\mathcal{GP}(c_\mu^i,k^i)$ for the GP prior on $d_i(x)$ in short.
We set the mean prior function to be constant, i.e. $\mu^i = c_\mu^i$.
Given observations $\bm y_N^i=[y_1^i,..,y_N^i]^T$ at locations $X_N=[ x_1,.., x_N]^T$ where
$y_j^i=d_i(x_j)$, the posterior distribution of $d_i(x)$ is given by
\begin{align}\label{eq:GP_mu}
\mu_N^i( x) &= c_\mu^i + k_N^i(x)^T {K_N^i}^{-1}(\bm y_N^i-\bm c_\mu^i)\\
k_N^i(x,x')&=k^i(x, x')- k_N^i(x)^T {K_N^i}^{-1}k_N^i(x')\\\label{eq:GP_sigma}
{\sigma_N^i}^2(x) &= k_N^i(x, x),
\end{align}
where $k_N^i(x)=[k^i(x_1, x),..,k^i( x_N,x)]^T$,
$k_N^i( x, x')$ is the posterior covariance, ${\sigma_N^i}^2(x)$ is the variance,
$K_N^i=[k^i(x,x')]_{x,x'\in X_N}$ is the positive definite covariance matrix
matrix and $\bm c_\mu^i$ is a vector of $N$ elements, each equal to $c_\mu^i$.
The posterior mean for the vector $d(x)$ is $\mu_N(x)=[\mu_N^0(x),..,\mu_N^n(x)]^\top$
and the variance $\sigma_N^2(x)=[{\sigma_N^0}^2(x),..,{\sigma_N^n}^2(x)]^\top$.

\subsection{GP-based bounding of nonlinearity}
\begin{figure}
\vspace{-0.25cm}
\end{figure}
\input{algorithms/QGP.alg}

The GP model provides a measure for the posterior
variance of $d(x)$, which is used to improve the
bound on the effect of the nonlinearity on the Lyapunov
decrease. Instead of using Lipschitz-based arguments,
as in Section \ref{sec:safe_sets_from_data},
we calculate a strict quadratic bound
on highly probable, worst-case realizations of the nonlinear
term $x^\top P d(x)$ for all
$x\in \mSafeSet^P(\gamma_i^2)\setminus\mSafeSet^P(\gamma_i^1)$,
given an interval $\Gamma_i$. The intervals $\Gamma_i$ in this section are not limited
to a dense data region $\mDataRegionDense$. 
Therefore one can drop the constraint \eqref{eq:intial_guess_P_0} in the computation of $P$
in \eqref{eq:intial_guess_P} and choose
$\bar \gamma = 1$ for the construction of the intervals \eqref{eq:intervals}.
By relying on the GP, the largest interval $\bar \gamma$
can be chosen independent of Assumption~\ref{ass:big_data} and
Remark~\ref{rem:ellipsoidal_data_region}.

Algorithm \ref{alg:calculationQGP} summarizes the calculation of the quadratic bound,
implementing the following idea. Let $f(x)$ be a function
that has to be quadratically upper bounded
by $x^\top Q x$ for all
$x\in \lbrace \mSafeSet^P(\gamma_i^2)\setminus\mSafeSet^P(\gamma_i^1)\rbrace$. 
In order to enforce the infinite dimensional constraint
$\forall x\in \lbrace \mSafeSet^P(\gamma_i^2)\setminus\mSafeSet^P(\gamma_i^1)\rbrace: f(x)\leq x^\top Q x$, we proceed iteratively by
starting with a finite approximation, which will be
improved until $f(x)\leq x^\top Q x$ holds for all $x\in\lbrace \mSafeSet^P(\gamma_i^2)\setminus\mSafeSet^P(\gamma_i^1)\rbrace$.

In line \ref{alg:calculationQGP_lambdaFcn} of Algorithm \ref{alg:calculationQGP},
the function $f$ is defined, which returns the maximum value of the
nonlinear term $x^\top P d(x)$ with a chosen probability, e.g. with $99.73\%$
by letting $c=3$.
Starting with a finite number
of samples of $f(x)$ for $x\in\mSafeSet^P(\gamma_i^2)\setminus\mSafeSet^P(\gamma_i^1)$ (lines \ref{alg:calculationQGP_initialX},\ref{alg:calculationQGP_initialY})
we compute an initial guess for a quadratic bound on 
$x^\top P d(x)$ given by $x^\top Q^1 x$ in
line \ref{alg:calculationQGP_initialQ}, where
$G(X,Y) = \argmin_{\tilde \mQuadBoundMatrix\in \mSetSymMat{n}}
		\sum_{(x_i,y_i)\in(X,Y)} \left(x_i^\top \tilde\mQuadBoundMatrix x_i-y_i\right)^2$
\text{s.t. for all } $(x_i,y_i)\in \mBoundSamples:~y_i \leq x_i^\top \tilde\mQuadBoundMatrix x_i$,
which yields a quadratic upper bound on $\lbrace y_i \rbrace_{i=1}^N$.

We search for potential violations of the current bound in
line \ref{alg:calculationQGP_violatingX} and add
it to the set of data points in lines \ref{alg:calculationQGP_newX}
and \ref{alg:calculationQGP_newY}. After that we update
the quadratic bound in line \ref{alg:calculationQGP_updateQ}.
The algorithm iterates until there is no violation.
This way for all $x\in \mSafeSet^P(\gamma_i^2)\setminus\mSafeSet^P(\gamma_i^1)$, the quadratic
form $x^\top Q_{GP}(\Gamma_i)x$ will be a bound on $x^\top P d(x)$
with high probability. 

\begin{remark}
	The optimization problem
	in line \ref{alg:calculationQGP_violatingX} is continuous
	(compare for example \cite[Chapter 2G]{dontchev2009implicit}),
	but non-convex.
	An alternative approach is to build a discretization of $\mSafeRing(\Gamma_i)$,
	which we denote by $\mDataSetX$, with grid size $\delta$, $|\mDataSetX|=N$, and evaluate
	the posterior mean and covariance $\mu_N(x), \sigma_N(x)$ for each
	$x\in \mDataSetX$.
	By selecting
	$p_k = x_k^\top P \mu_N(x_k) + \beta_N \sum_{i=1}^n \sigma_N^i(x) +
	\delta L_{x^\top P d(x)}(\mSafeRing(\Gamma_i))$ with
	$\beta_N$ as defined in \cite[Lemma 3]{Berkenkamp2017SafeRL}, the bound
	$Q_{GP}$ can be approximated using \eqref{eq:quad_bound} as a
	convex optimization problem.
\end{remark}

For the second step of calculating a safe set size and controller, we can simply use the
bound $Q_{GP}(\Gamma_i)$ instead of $Q(\Gamma_i)$ in \eqref{eq:general_opt} in order
to obtain a set, which is safe with the selected probability.
By construction, $Q_{GP}(\Gamma_i)$ will
be less conservative, or equal than $Q(\Gamma_i)$. This is due to the fact that we put a prior
on the unknown nonlinearity $d(x)$, which allows for Bayesian inference and therefore
improved extra- and interpolation based on the data, as opposed to using estimates based
on Lipschitz continuity. In Figure \ref{fig:example} (left), the benefit of using
$Q_{GP}$ over the Lipschitz based bound $Q$ is illustrated.
 Moreover, the main advantage is that we do
\emph{not} require particular assumptions on the data and the safe set is not
limited to a subset of $\mDataRegionDense$, as it is the case in Lemma~\ref{lem:quad_bound}.

\subsection{Numerical example: Exploration}\label{subsec:exploration}
Consider a nonlinear system of the form \eqref{eq:general_system} with
$A = \begin{psmallmatrix} -1 &2 \\ -3& 4 \end{psmallmatrix}$,
$B=\begin{psmallmatrix} 0.5 \\ -2 \end{psmallmatrix}$,
$d(x) = \begin{psmallmatrix} 0.5x_1^2\sin(6x_1) \\ -0.8x_2^3 \end{psmallmatrix}$,
input constraints $|u|\leq 4$, and state constraints $|x|\leq 4$. We use a squared exponential
kernel as defined in \cite{rasmussen2006gaussian}
with ${\sigma_f^1}={\sigma_f^2}=0.05$ and $l^1=l^2=0.2$.
Given initial data in $[-0.2,0.2]^2$ with $\delta = 0.05$,
we solve \eqref{eq:general_opt} using the high probability ($c=3$, $99.73\%$) bound $Q_{GP}(\Gamma_i)$
with an initial $P$ obtained via \eqref{eq:intial_guess_P} and
$\Gamma_i = [1-i0.1-0.1,1-i0.1]$ for $i=1,2,..,8$,
i.e. we start with $i=1$ and iterate $i=2,3,..$ until we find a feasible solution.
We assume that the desired control input $\bar u(t)$ is given by the policy gradient with signed derivative (PGSD) algorithm
(see \cite{kolter2009policy}), which is a policy search RL method without any safety guarantees. 
During closed-loop operation under the safe control law
\eqref{eq:safe_control_law} we collect data $\mDataSet$. Every $0.2$ seconds we recompute
the safe level set, i.e. $\gamma(t)$, where the safe set size converges after $2.5$ seconds.
In Figure~\ref{fig:example_exploration} the evolution of the safe set size (volume
of the ellipse) is shown as well as the distance of the system state to the origin,
which has to be minimized by the PGSD learning based control law.
The unsafe RL input is `overwritten' three times indicated by the grey lines
to ensure safety, until it begins to converge.

\begin{figure}[t]
    \centering
	\input{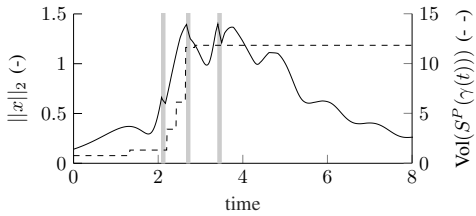}
	\caption{Safe RL using the proposed safety framework together with the PGSD \cite{kolter2009policy} algorithm.
	The distance of the state to the origin, as well as the volume of the
	safe set is shown over time.
	Grey-shaded time points depict application of the safe control law $u_\mSafeSet$.}
    \label{fig:example_exploration}
\end{figure}

\section{Conclusion}\label{sec:conclusion}
This paper presents a safety framework that
allows to enhance arbitrary learning-based and unsafe control strategies,
for nonlinear and potentially larger scale systems with safety certificates. The nonlinearity
is assumed to be unknown and we only require a possibly inaccurate linear system
model and observations of the system. A key feature is that the proposed method directly
exploits the available data, without the need of an additional
learning mechanism. By relying on convex optimization problems, the
proposed method is scalable with respect to the system dimension and number of data points.
In order to reduce conservatism of the safe set calculations, the approach was
extended using a Gaussian process model as prior on the nonlinearity.
This modification enables safe exploration and thereby iterative computation
of the safe set during closed-loop operation.
The results were demonstrated using several numerical example problems,
showing that the safety framework can be used to certify arbitrary and in
particular learning-based controllers.

\bibliography{bibliography.bib}
\bibliographystyle{IEEEtran}

\end{document}

%% file: definitions/definitionsKim.tex
\theoremstyle{plain}
\newtheorem{theorem}{Theorem}[section]
\theoremstyle{plain}

\theoremstyle{plain}
\theoremstyle{definition}
\newtheorem{definition}[theorem]{Definition}
\theoremstyle{plain}
\newtheorem{lemma}[theorem]{Lemma}
\theoremstyle{plain}

\theoremstyle{plain}
\newtheorem{remark}[theorem]{Remark}
\newtheorem{assumption}[theorem]{Assumption}

\newcommand{\RR}{\mathbb{R}}

\newcommand{\mSetSymMat}[1]{S^{#1}}
\newcommand{\mSetPosSemSymMat}[1]{S^{#1}_+}
\newcommand{\mSetPosSymMat}[1]{S^{#1}_{++}}
\newcommand{\mIntInt}[2]{\mathcal I_{[#1, #2]}}
\newcommand{\mIntGeq}[1]{\mathcal I_{\geq #1}}
\newcommand{\mDataSet}{\mathcal D}
\newcommand{\mDataSetX}{\mathcal D_x}
\newcommand{\mDataSetY}{\mathcal D_y}
\newcommand{\mBoundSamples}{\mathcal P}
\newcommand{\mNorm}[2]{\left\lVert#1\right\rVert_{#2}}
\newcommand{\mDefFunction}[3]{#1: #2 \rightarrow #3}

\newcommand{\mBall}[2]{\mathcal B_{#1}(#2)}

\newcommand{\XX}{\mathbb{X}}
\newcommand{\UU}{\mathbb{U}}
\newcommand{\mSafeSet}{\mathcal{S}}
\newcommand{\mSafeRing}{\mathcal {R}}
\newcommand{\mSafeRingIndices}{\mathbb I_\mSafeRing}

\newcommand{\mDataRegion}{\mathcal A}
\newcommand{\mDataRegionDense}{\mDataRegion_\delta}
\newcommand{\mQuadBoundMatrix}{Q}
\newcommand{\mQuadBoundMatrixFunction}[1]{\mQuadBoundMatrix(#1)}
\newcommand{\mQuadBoundMatrixFunctionIteration}[2]{\mQuadBoundMatrix^{#1}(#2)}

\newcommand{\mGridFunction}{\Delta_{\mDataSetX}}

\DeclareMathOperator*{\argmin}{argmin}